\documentclass[letterpaper,10pt,conference]{ieeeconf}

\IEEEoverridecommandlockouts

\usepackage{amsmath}
\usepackage{amssymb}
\usepackage{comment}
\usepackage{amsfonts}

\usepackage{amsthm}
\usepackage{cite}
\usepackage{graphicx}
\usepackage{hyperref}
\usepackage{booktabs}

\usepackage[ruled,vlined]{algorithm2e}

\usepackage[color=green!40]{todonotes}

\usepackage{mathrsfs}

\theoremstyle{plain}

\newtheorem{remark}{Remark}

\newtheorem{lemma}{Lemma}

\theoremstyle{plain}
\newtheorem{theorem}{Theorem}
\theoremstyle{definition}
\newtheorem{definition}{Definition}
\theoremstyle{definition}

\theoremstyle{definition}
\newtheorem{problem}{Problem}

\title{\LARGE \bf
Mobility Equity and Economic Sustainability Using Game Theory
}

\author{Ioannis Vasileios Chremos, \textit{Student Member, IEEE}, and Andreas A. Malikopoulos, \textit{Senior Member, IEEE}%
\thanks{This work was supported by NSF under Grants CNS-2149520 and CMMI-2219761.}%
\thanks{The authors are with the Department of Mechanical Engineering, University of Delaware, Newark, DE 19716 USA. {\tt\small{\{ichremos,andreas\}@udel.edu.}}}%
}

\begin{document}

\maketitle
\thispagestyle{empty}
\pagestyle{empty}

\begin{abstract}

In this paper, we consider a multi-modal mobility system of travelers each with an individual travel budget, and propose a game-theoretic framework to assign each traveler to a ``mobility service" (each one representing a different mode of transportation). We are interested in equity and sustainability, thus we maximize the worst-case revenue of the mobility system while ensuring ``mobility equity," which we define it in terms of accessibility. In the proposed framework, we ensure that all travelers are truthful and voluntarily participate under informational asymmetry, and the solution respects the individual budget of each traveler. Each traveler may seek to travel using multiple services (e.g., car, bus, train, bike). The services are capacitated and can serve up to a fixed number of travelers at any instant of time. Thus, our problem falls under the category of many-to-one assignment problems, where the goal is to find the conditions that guarantee the stability of assignments. We formulate a linear program of maximizing worst-case revenue under the constraints of mobility equity, and we fully characterize the optimal solution.

\end{abstract}

\section{Introduction}

%   \subsection{Motivation}

Commuters in big cities have continuously experienced the frustration of congestion and traffic jams. Travel delays, accidents, and road altercations have consistently impacted the economy, society, and the natural environment in terms of energy and pollution \cite{colini_baldeschi2017}. One of the pressing challenges of our time is the increasing demand for energy, which requires us to make fundamental transformations in how our societies use and access transportation. \emph{Emerging mobility systems} (EMS) (e.g., connected and automated vehicles (CAVs), shared mobility, ride-hailing, on-demand mobility services) are expected to eliminate congestion and significantly increase mobility efficiency in terms of energy and travel time. Several studies have shown the benefits of EMS to reduce energy and alleviate traffic congestion in a number of different transportation scenarios (see references therein \cite{zhao2019enhanced,Malikopoulos2020}). However, recently it was shown that when daily commuters were offered a convenient and affordable taxi service for their travels, a change of behavior was noticed, namely these commuters ended up using taxi services more often compared to when they drove their own car \cite{harb2018}. Other studies \cite{bissell2018,singleton2018} have shown similar results which is an indication that EMS could potentially affect people's tendency to travel and incentivize them to use cars more frequently which potentially can also lead to a shift away from public transit. 

Wide accessibility to transportation in EMS can be impacted by the socioeconomic background of the travelers, i.e., whether a traveler can afford it. For example, travelers with a low-income background may be unable to use any or all transportation travel options available in a city. Thus, our approach, in this paper, is to study the game-theoretic interactions of travelers seeking to travel in a multi-modal mobility system, where each traveler has a unique travel budget. We adopt the Mobility-as-a-Service (MaaS) concept, which is a system of multi-modal mobility that handles user-centric information and provides travel services (e.g., navigation, location, booking, payment) to a number of travelers. So, travelers are expected to report their preferences to a central authority. The goal is to guarantee mobility as a seamless service across all modes of transportation accessible to all in a socially-efficient and fair way.

%   \subsection{Literature Review}

One of the standard approaches to alleviate congestion in a transportation system has been the management of demand size due to the shortage of space availability and scarce economic resources in the form of congestion pricing. Such an approach focuses primarily on intelligent and scalable traffic routing, in which the objective is to guide and coordinate users in path-choice decision-making. For example, one computes the shortest path from a source to a destination regardless of the changing traffic conditions \cite{schmitt2006}. Interestingly, by adopting a game-theoretic approach, advanced systems have been proposed to assign users concrete routes or minimize travel time and study the Nash equilibria under different tolling mechanisms \cite{salazar2019,Chremos2020SocialDilemma,Chremos2020MechanismDesign,chremos2020SharedMobility}. Partly related to our work are matching models which describe markets in which there are agents of disjoint groups and have preferences regarding the ``goods" of the opposite agent they associate with. Two-sided matching with transfers have been modeled as assignment problems where one entity (e.g., a firm) needs to pay salaries to individuals (e.g., workers) \cite{shapley1971,kelso1982}. Tasks-matching problems under a wide range of constraints have been reported in \cite{Kojima2020}. A wider literature on matching under constraints can be found in \cite{Aziz2022}. Notable examples are mechanisms that assign students to schools, interns to hospitals, worker-firm contracts, or travelers to vehicles \cite{abdulkadiroglu2005,westkamp2013,chremos2021MobilityGame}. It is easy to see that matching markets are quite practical as they offer insights into the more general economic and behavioral real-life situations. These examples are all centralized approaches of determining who gets assigned to whom at what cost and benefit. Given the natural usefulness of matching markets, various extensions of the assignment game have been developed focusing either on different behavioral settings or information structures \cite{demange1986,sotomayor1992,anshelevich2013,chremos2020MobilityMarket}. Assignment games and matching markets have also been studied extensively in auction theory \cite{Myerson1981,bandi2014}.

%   \subsection{Contribution of the Paper}

The main contribution of this paper is the game-theoretic development of framework to study the socioeconomic interactions of travelers in a multi-modal mobility system. We focus our analysis in the \emph{economic sustainability} and \emph{mobility equity} of the mobility system. We offer a game-theoretic definition of equity based on accessibility, i.e., we ensure our framework satisfies the following properties: truthfulness, voluntary participation, and budget fairness. In particular, we formulate our problem as a linear program to compute the assignments between travelers and mobility services that maximize the worst-case revenue of the system. We consider that information is asymmetric, i.e., a social planner has no knowledge of the individual travelers' valuations of the services. We provide a pricing mechanism and show how we can elicit the private information truthfully (Theorem \ref{thm:IC}) while ensuring budget fairness (Theorem \ref{thm:ME}). We also show that every traveler voluntary participates in our proposed framework (Theorem \ref{thm:IR}).

%   \subsection{Organization of the Paper}

The remainder of the paper is structured as follows. In Section \ref{Section:Formulation}, we present the mathematical formulation of the proposed game-theoretic framework. In Section \ref{Section:Analysis&Properties}, we derive the theoretical properties of our framework, and finally, in Section \ref{Section:Conclusion}, we discuss the implementation of the proposed framework, draw conclusions and future directions.

\section{Modeling Framework}
\label{Section:Formulation}

\subsection{Mathematical Formulation}

We consider a mobility system where $I \in \mathbb{N}_{\geq 2}$ travelers, indexed by $i \in \mathcal{I}$, $|\mathcal{I}| = I$, are interested in using $J \in \mathbb{N}$ mobility services, indexed by $j \in \mathcal{J}$, made available by a \emph{social planner}. In addition, we expect $I < J$.
%central city authority
Any $j \in \mathcal{J}$ represents the service that can be offered to traveler $i$. For example, a taxicab service, say some $j \in \mathcal{J}$, can satisfy the travel needs of up to five travelers; thus, any service $j$ can be divided to multiple travelers based on the service $j$'s physical capacity. Each traveler $i \in \mathcal{I}$ has a private valuation $v_{i j}$ associated with each of the services $j \in \mathcal{J}$, which is not known to the social planner.

Travelers are constrained by a \emph{travel budget}, $b_i \in \mathbb{R}_{\geq 0}$, for any traveler $i \in \mathcal{I}$. Thus, we can only charge travelers payments that do not exceed their individual budgets. We write $\mathcal{B} = \{b_1, b_2, \dots, b_I\}$. For the purposes of this work, we assume that the budgets of each traveler are known to the social planner. Our reasoning here is twofold: A probabilistic distribution for unknown private budgets leads to an impossibility result for socially-efficient mechanisms \cite{Dobzinski2008,Dobzinski2012}. In addition, based on transportation literature, it is reasonable to expect travelers to submit their travel budget on a mobility app \cite{Goodwin1981,Tussyadiah2006}.

For each service $j \in \mathcal{J}$, we model the social planner's beliefs on the realization of the private valuations for service $j$ as real values from some subset of real values.

\begin{definition}
    For each traveler $i \in \mathcal{I}$, the traveler $i$'s valuation profile of all mobility services is $v_i = (v_{i 1}, v_{i 2}, \dots, v_{i J})$, $v_{i j} \in \mathbb{R}$. We write $v_{- i j} = (v_{1 j}, \dots, v_{(i - 1) j}, v_{(i + 1) j}, \dots, v_{I j})$ for the valuation profile of all travelers except $i$ for service $j$ and denote by $v_{- i} = (v_{- i 1}, \dots, v_{- i J})$ the profile of valuations of all services of all travelers except traveler $i$. Then, $v = (v_i, v_{- i}) \in \mathcal{V} \subset \mathbb{R} ^ {I \times J}$ is the valuation profile of all travelers for all mobility services.
\end{definition}

For an arbitrary traveler $i$, the valuation $v_{i j}$ can represent the realization of a satisfaction function that captures, for example, the maximum amount of money that traveler $i$ is willing to pay for mobility service $j$.

Travelers may use multiple services to satisfy their travel needs, i.e., to reach their destination, via a smartphone app. The social planner then compiles all travelers' origin-destination requests and other information (e.g., preferred travel time, value of time, and maximum willingness-to-pay) in order to provide a travel recommendation to each traveler. However, we consider that the travelers' budgets are known to the social planner as it is reasonable to expect travelers to submit their travel budget on a mobility app \cite{Goodwin1981,Tussyadiah2006}. The travelers' valuations for each different mobility service are considered private information as we cannot expect travelers to provide truthfully their preferences for any service.

The allocation of the finite number of mobility services to travelers can be described by a vector of binary variables.

\begin{definition}\label{defn:traveler_service_assign}
    The \emph{traveler-service assignment} is a vector $\mathbf{a} = (a_{i j}(v))_{{i \in \mathcal{I}}, j \in \mathcal{J}}$, where $a_{i j}$ is a binary variable, i.e.,
        \begin{equation}\label{EQN:binary-variable}
            a_{i j}(v) =
                \begin{cases}
                    1, \; & \text{if $i \in \mathcal{I}$ is assigned to $j \in \mathcal{J}$}, \\
                    0, \; & \text{otherwise}.
                \end{cases}
        \end{equation}
\end{definition}

Note that the assignment $a_{i j}(v)$ between traveler $i$ and service $j$ depends on the valuation $v_i$ of traveler $i$ and the valuations of all other travelers, i.e., $v_{- i}$.

It is possible in our framework for a traveler to reject all assignments with any service. However, we show in Theorem \ref{constraint:IR} how to avoid such an outcome by providing the right incentives to travelers to use at least one service. Naturally, though, each service can accommodate up to a some number of travelers, different for each type of services. So, we expect the ``physical traveler capacity" of each service to vary significantly. 

\begin{definition}
    Each service $j \in \mathcal{J}$ is associated with a current \emph{traveler capacity}, denoted by $\varepsilon_j \in \mathbb{N}$ and $\varepsilon_j \leq \bar{\varepsilon}_j$, where $\bar{\varepsilon}_j$ denotes the maximum traveler capacity of service $j$.
\end{definition}

For example, a bus can provide travel services to a hundred travelers (seated and standing) compared to a bike-sharing company's bike (since one traveler per bike).

\begin{definition}
    A \emph{feasible} assignment is a vector $\mathbf{a} = (a_{i j})_{{i \in \mathcal{I}}, j \in \mathcal{J}}$, $a_{i j} \in \{0, 1\}$ that satisfies 
        \begin{align}
            \sum_{j \in \mathcal{J}} a_{i j}(v) & \leq \delta_i, \quad \forall i \in \mathcal{I}, \quad \forall v \in \mathcal{V}, \label{Constraint:Problem1-First} \\
            \sum_{i \in \mathcal{I}} a_{i j}(v) & \leq \bar{\varepsilon}_j, \quad \forall j \in \mathcal{J}, \quad \forall v \in \mathcal{V}, \label{Constraint:Problem1-Second}
        \end{align}
    where \eqref{Constraint:Problem1-First} ensures that each traveler $i \in \mathcal{I}$ is assigned to at most $\delta_i \in \mathbb{N}$ mobility service $j \in \mathcal{J}$, and \eqref{Constraint:Problem1-Second} ensures that the traveler capacity of each service $j$ is not exceeded while it is shared by multiple travelers.
\end{definition}

We represent the preferences of each traveler with a utility function consisted of two parts: the traveler's valuation of the mobility outcome and the associated payment required for the realization of that outcome. In other words, any traveler is expected to pay a fare/toll for using the mobility service.

\begin{definition}\label{defn:utility_function}
    Each traveler $i$'s preferences are summarized by a utility function $u_i : \mathcal{V} \times \mathbb{R} \to \mathbb{R}$ that determines the monetary value of the overall payoff realized by traveler $i$ from their assignment to service $j$. Thus, traveler $i$ receives a total utility
        \begin{equation}\label{eqn:traveler_utility}
            u_i((v_i, v_{- i}), p_i) = \sum_{j \in \mathcal{J}} v_{i j} a_{i j}(v_i, v_{- i}) - p_i(v_i, v_{- i}),
        \end{equation}
    where $p_i \in \mathbb{R}$ denotes traveler $i$'s mobility payment.
\end{definition}

Note that each traveler's goal is to choose a strategy that maximizes their utility only. Next, we present the definition of ``mobility equity" of our game-theoretic framework.

\begin{definition}\label{defn:equity}
    A mobility system $\langle \mathcal{I}, \mathcal{J}, \mathcal{V}, (u_i)_{i \in \mathcal{I}}, (p_i)_{i \in \mathcal{I}} \rangle$ admits an equilibrium that is \emph{mobility equitable} if (i) travelers truthfully report their private information, (ii) travelers voluntarily participate, and (iii) travelers can afford travel.
\end{definition}

Next, we formally define the relation that ensures ``economic sustainability" for our game-theoretic framework.

\begin{definition}\label{defn:SUSTAIN}
    Let linear function $w : \mathcal{V}\times \mathbb{R} \to \mathbb{R}_{\geq 0}$ that depends on the valuations, assignments, and individual budgets of all the travelers denote the worst-case revenue. Mathematically, we have
        \begin{equation}\label{eqn:mobility-equity}
            \sum_{i \in \mathcal{I}} \sum_{j \in \mathcal{J}} w_i(v_{i j}, a_{i j}, b_i) \leq \sum_{i \in \mathcal{I}} p_i(v), \quad \forall v \in \mathcal{V}.
        \end{equation}
\end{definition}

Our intuition behind Definition \ref{defn:SUSTAIN} is conceptually based on what the United Nations Development Programme has developed as part of their Sustainable Development Goals. In particular, our goal in this work is to ensure long-term economic growth in the worst possible cases (thus, maximizing \eqref{eqn:mobility-equity}) under the constraints of Definition \ref{defn:equity}. We now define the constraints that ensure mobility equity in our framework's solutions based on Definition \ref{defn:equity}. 

\begin{definition}\label{defn:IC}
    For the travelers to have no incentive to misreport their valuations to the social planner, we need
        \begin{multline}\label{constraint:IC}
            \sum_{j \in \mathcal{J}} v_{i j} a_{i j}(\tilde{v}_i, v_{- i}) - p_i(\tilde{v}_i, v_{- i}) \\
            - \sum_{j \in \mathcal{J}} v_{i j} a_{i j}(v_i, v_{- i}) + p_i(v_i, v_{- i}) \leq 0,
        \end{multline}
    for all $v = (v_i, v_{- i}) \in \mathcal{V}$, any $\tilde{v}_i \in \mathcal{V}$, and for all travelers $i \in \mathcal{I}$ using any mobility service $j \in \mathcal{J}$. We call $\tilde{v}_i$ traveler $i$'s reported valuation that deviates from the true valuation $v_i$. If \eqref{constraint:IC} holds, we say that the mechanism induces \emph{truthfulness}.
\end{definition}

\begin{definition}\label{defn:IR}
    The travelers in the mobility system \emph{voluntarily participate} (VP) if, for any traveler $i \in \mathcal{I}$,
        \begin{equation}\label{constraint:IR}
            p_i(v_i, v_{- i}) \leq \sum_{j \in \mathcal{J}} v_{i j} a_{i j}(v_i, v_{- i}), \quad \forall v \in \mathcal{V}.
        \end{equation}
    We say then that the proposed mechanism induces voluntary participation from all travelers.
\end{definition}

\begin{definition}\label{defn:ME}
    The mechanism induces on individual level \emph{budget fairness} (BF), if for any traveler $i \in \mathcal{I}$, we have
        \begin{equation}\label{constraint:ME}
            p_i(v) \leq b_i, \quad \forall v \in \mathcal{V}.
        \end{equation}
\end{definition}

\subsection{The Optimization Problem}

\begin{problem}\label{pro:primal}
    The maximization problem is formulated as
        \begin{gather}
            \max_{a_{i j}} \sum_{i \in \mathcal{I}} \sum_{j \in \mathcal{J}} w_i(v_{i j}, a_{i j}, b_i), \label{Equation:Problem1-ObjectiveFunction} \\
            \text{subject to:} \notag \;
            \eqref{Constraint:Problem1-First},
            \eqref{Constraint:Problem1-Second},
            \eqref{eqn:mobility-equity},
            \eqref{constraint:IC},
            \eqref{constraint:IR},
            \eqref{constraint:ME},
        \end{gather}
    where $a_{i j} \in \{0, 1\}$ for all $i \in \mathcal{I}$ and all $j \in \mathcal{J}$.
\end{problem}

We note here that Problem \ref{pro:primal} is a special case of the many-to-many assignment problem that is known to be very hard to solve analytically. Thus, we relax the integer constraint and focus our analysis on deriving the optimal solutions of a linearized version of Problem \ref{pro:primal}. Thus, we introduce a non-negativity constraint variable as follows.

\begin{problem}\label{prob:linear}
    The linear program formulation is
        \begin{gather}
            \max_{a_{i j}} \sum_{i \in \mathcal{I}} \sum_{j \in \mathcal{J}} w_i(v_{i j}, a_{i j}, b_i) \label{Equation:Problem2-ObjectiveFunction} \\
            \text{subject to:} \notag \;
            \eqref{Constraint:Problem1-First},
            \eqref{Constraint:Problem1-Second},
            \eqref{eqn:mobility-equity}, \eqref{constraint:IC},
            \eqref{constraint:IR},
            \eqref{constraint:ME},
            \text{ and} \\
            a_{i j} \geq 0, \quad \forall i \in \mathcal{I}, \quad \forall j \in \mathcal{J}, \label{Constraint:Problem2-Third}
        \end{gather}
    where \eqref{Constraint:Problem2-Third} transforms the (binary) assignment problem to a (continuous) linear program.
\end{problem}

\begin{remark}
    Intuitively, in Problem \ref{prob:linear}, if $a_{i j}>0$ it implies that traveler $i \in \mathcal{I}$ is assigned to service $j \in \mathcal{J}$.
\end{remark}

Problem \ref{prob:linear} is a constrained linear maximization problem that admits at least one solution under certain conditions. A solution of Problem \ref{prob:linear} ensures the assignments between the travelers and services are mobility equitable and economic sustainable. We maximize the worst-case revenue of the mobility system under the constraints of truthfulness, VP, and BF. Next, inspired from Myerson's auction \cite{Myerson1981} and the Vickrey-Clarke-Groves (VCG) auction mechanism, we introduce two key variables that can help us solve Problem \ref{prob:linear}, i.e., \emph{nominal assignments} and \emph{reservation payments}.

\begin{definition}
    For any traveler $i \in \mathcal{I}$, there is a reservation payment for each mobility service $j \in \mathcal{J}$, denoted by $r_{i j} \in \mathbb{R}_{\geq 0}$, representing the minimum necessary mobility payment of traveler $i$ to get assigned to mobility service $j$.
\end{definition}

\begin{definition}\label{defn:final_assignment}
    The final assignment $a_{i j}(v)$ evaluated at the realized valuation profile $v \in \mathcal{V}$ is computed as the sum of the nominal assignment $\bar{a}_{i j}$ and the adapted assignment $\tilde{a}_{i j}(v)$, i.e., we have $a_{i j}(v_i, v_{- i}) = \bar{a}_{i j} + \tilde{a}_{i j}(v_i, v_{- i})$.
\end{definition}

\section{Analysis and Properties of the Mechanism}
\label{Section:Analysis&Properties}

In this section, we show that the proposed mechanism satisfies the desired properties of mobility equity (Definition \ref{defn:equity}). We start by stating the dual program of Problem \ref{prob:linear}.

\begin{lemma}\label{lem:dual_problem}
    The dual problem of Problem \ref{prob:linear} is
        \begin{gather}
            \min \sum_{v \in \mathcal{V}} \left[ \sum_{i \in \mathcal{I}} \xi_1 ^ i(v) \delta_i + \sum_{j \in \mathcal{J}} \xi_2 ^ j (v) \bar{\varepsilon}_j + \sum_{i \in \mathcal{I}} \xi_4 ^ i (v) b_i \right] \\
            \text{subject to: } \xi_1 ^ i (v) + \xi_2 ^ j (v) + \sum_{\tilde{v}_i \in \mathcal{V}} \tilde{v}_{i j} \xi_4 ^ i (v, \tilde{v}_i) \notag \\
            - v_{i j} \sum_{\tilde{v}_i} \xi_4 ^ i (v, \tilde{v}_i) - v_{i j} \xi_5 ^ i (v) \geq 0, \quad \forall v \in \mathcal{V}, \\
            \sum_{\tilde{v}_i \in \mathcal{V}} \xi_4 ^ i (v, \tilde{v}_i) - \sum_{\tilde{v}_i \in \mathcal{V}} \xi_4 ^ i (v_{- i}, \tilde{v}_i, v_i) - \xi_3 (v) \notag \\
            + \xi_5 ^ i(v) + \xi_6 ^ i (v) = 0, \quad \forall v \in \mathcal{V}, \\
            \sum_{v \in \mathcal{V}} \xi_3 (v) = 1, \\
            \xi_1(v), \xi_2 ^ j (v), \xi_3, \xi_4 ^ i (v, \tilde{v}_i), \xi_5 ^ i (v), \xi_6 ^ i (v) \geq 0,
        \end{gather}
    where $\xi_1 ^ i, \xi_2 ^ j, \xi_3, \xi_4 ^ i, \xi_5 ^ i$, and $\xi_6 ^ i$ are the dual variables for constraints \eqref{Constraint:Problem1-First}, \eqref{Constraint:Problem1-Second}, \eqref{eqn:mobility-equity}, \eqref{constraint:IC}, \eqref{constraint:IR}, and \eqref{constraint:ME}, respectively.
\end{lemma}

\begin{proof}
    The computations here are straightforward, hence we omit them due to space limitations.
\end{proof}

Based on Lemma \ref{lem:dual_problem}, we can now compute the nominal assignments and reservation payments as follows: we formulate the optimization problem for the assignments
    \begin{gather}
        \max_{a_{i j}} \sum_{i \in \mathcal{I}} \sum_{j \in \mathcal{J}} v_{i j} a_{i j} \label{Equation:Problem4-ObjectiveFunction} \\
        \text{subject to:} \notag \;
        \eqref{Constraint:Problem1-First},
        \eqref{Constraint:Problem1-Second},
        \text{ and} \\
        v_{i j} a_{i j}(\tilde{v}_i, v_{- i}) - v_{i j} a_{i j}(v_i, v_{- i}) \leq 0, \\
        \forall \tilde{v}_i \in \mathcal{V}, \quad \forall i \in \mathcal{I}, \quad \forall j \in \mathcal{J}, \notag \\
        \sum_{j \in \mathcal{J}} v_{i j} a_{i j}(v_i, v_{- i}) \leq b_i, \quad \forall i \in \mathcal{I}, \\
        a_{i j} \geq 0, \quad \forall i \in \mathcal{I}, \quad \forall j \in \mathcal{J}, \label{Constraint:Problem4-Third}
    \end{gather}
where solving \eqref{Equation:Problem4-ObjectiveFunction} gives us the valuation profile $v ^ {\text{worst}} = (v_{i j} ^ {\text{worst}})_{i \in \mathcal{I}, {j \in \mathcal{J}}}$ at the worst case and the associated nominal assignment $\bar{a}_{i j}$. Next, we derive $\xi_1 ^ i$, $\xi_2 ^ j$, $\xi_5 ^ i$, and $\xi_6 ^ i$ from Lemma \ref{lem:dual_problem} and then compute $r_{i j} = \xi_1 ^ i + \xi_2 ^ j + \xi_5 ^ i v_{i j} ^ {\text{worst}} + \xi_6 ^ i v_{i j} ^ {\text{worst}}$. The next step now is to present the pricing mechanism for any traveler $k \in \mathcal{I}$ of our proposed framework. But first, we define $\gamma_{i j} = \arg \min_{\tilde{v} \in \mathcal{V}} \sum_{j \in \mathcal{J}} \bar{a}_{i j} \tilde{v}_{i j}$. Then, we have
    \begin{multline}\label{eqn:mobility_payment}
        p_k(v) = \sum_{j \in \mathcal{J}} \tilde{a}_{k j} (v) r_{k j} + \sum_{j \in \mathcal{J}} \bar{a}_{k j} r_{k j} - \sum_{j \in \mathcal{J}} \bar{a}_{k j} \xi_5 ^ k \gamma_{k j} \\
        + \sum_{i \in \mathcal{I} \setminus \{k\}} \sum_{j \in \mathcal{J}} \tilde{a}_{i j; k} (v_{- k}) (v_{i j} - r_{i j}) \\
        - \sum_{i \in \mathcal{I} \setminus \{k\}} \sum_{j \in \mathcal{J}} \tilde{a}_{i j} (v) (v_{i j} - r_{i j}),
    \end{multline}
where $\tilde{a}_{i j; k}$ represents a ``temporary" assignment of travelers to mobility services expecting traveler $k$ (we see how to estimate this variable in \eqref{eqn:estimate:a^v-k}). We formally present how to compute such an assignment in Theorem \ref{thm:ME}.
The term $\sum_{i \in \mathcal{I}} \sum_{j \in \mathcal{J}} \tilde{a}_{i j} (v) (v_{i j} - r_{i j})$ represents the ``social welfare" of all travelers based on the valuations of each mobility service $j$ and the reservation mobility payments $r_{i j}$. We motivate our mobility pricing mechanism \eqref{eqn:mobility_payment} as follows: with the help of the reservation payments we parameterize the totals of social welfare in terms of the travelers' valuations. So, the first three terms capture the parameterized social welfare for all services from the point of view of one traveler. Then the other two terms represent the social welfare excluding traveler $k$'s contribution.
Using these reservation payments, we then introduce a mobility payment $p_k$ for traveler $k$ that charges the minimum required payment for traveler $k$ to get assigned to mobility service $j$ while keeping all other travelers' reported valuations fixed. Next, we show that our mechanism induces truthfulness from all travelers.

\begin{theorem}\label{thm:IC}
    The proposed framework induces all travelers to report their valuations $v \in \mathcal{V}$ truthfully to the social planner under the pricing mechanism \eqref{eqn:mobility_payment}.
\end{theorem}

\begin{proof}
    Consider traveler $k$ with a true valuation $v_{k j}$ for each service $j \in \mathcal{V}$. By reporting $v_{k j} '$, traveler $k$ is assigned service $j$ with $\tilde{a}_{k j}(v_k ', v_{- k})$. We formulate the optimization problem $(\tilde{a}_{i j}(v))_{i \in \mathcal{I}, j \in \mathcal{J}} = \arg \max_{\tilde{a} \in \mathcal{A}} \sum_{i \in \mathcal{I}} \sum_{j \in \mathcal{J}} \tilde{a}_{i j} (v_{i j} - r_{i j})$, where $\mathcal{A}$ is the set of positive values for $\tilde{a}$ that satisfies the following two constraints:
        \begin{align}
            \sum_{i \in \mathcal{I}} \tilde{a}_{i j} & \leq 1 - \sum_{i \in \mathcal{I}} \bar{a}_{i j}, \quad \forall j \in \mathcal{J}, \\
            \sum_{j \in \mathcal{J}} \tilde{a}_{i j} \tilde{v}_{i j} & \leq b_i - \sum_{j \in \mathcal{J}} \bar{a}_{i j} r_{i j} + \sum_{j \in \mathcal{J}} \bar{a}_{k j} \xi_5 ^ i \gamma_{i j}, \label{eqn:necessary_condition}
        \end{align}
    where \eqref{eqn:necessary_condition} must hold for all $i \in \mathcal{I}$, for all $\tilde{v} \in \mathcal{V}$, and $\gamma_{i j} = \arg \min_{\tilde{v} \in \mathcal{V}} \sum_{j \in \mathcal{J}} \bar{a}_{i j} \tilde{v}_{i j}$. Since $\mathcal{A}$ does not depend on any specific valuation profile, we have $(\tilde{a}_{i j}(v_k ', v_{- k}))_{i j} \in \mathcal{A}$. Thus, we have $\sum_{i \in \mathcal{I}} \sum_{j \in \mathcal{J}} \tilde{a}_{i j}(v_k, v_{- k}) (v_{i j} - r_{i j}) \geq \sum_{i \in \mathcal{I}} \sum_{j \in \mathcal{J}} \tilde{a}_{i j}(v_k ', v_{- k}) (v_{i j} - r_{i j})$. Using Definition \ref{defn:utility_function}, we now compare the utilities of traveler $k$ under the two different valuations. So, we have $u_k(v_k, v_{- k}) = \sum_{j \in \mathcal{J}} a_{k j}(v_k, v_{- k}) v_{k j} - p_k(v_k, v_{- k})$, which, by Definition \ref{defn:final_assignment} and \eqref{eqn:mobility_payment}, we can expand as follows
        \begin{multline}
            u_k(v_k, v_{- k}) = \sum_{j \in \mathcal{J}} \tilde{a}_{k j}(v_k, v_{- k}) v_{k j} + \sum_{j \in \mathcal{J}} \bar{a}_{k j} v_{k j} \\
            - \sum_{j \in \mathcal{J}} \tilde{a}_{k j} r_{k j} - \sum_{j \in \mathcal{J}} \bar{a}_{k j} r_{k j} + \sum_{j \in \mathcal{J}} \bar{a}_{k j} \xi_5 ^ k \gamma_{k j} \\
            - \sum_{i \in \mathcal{I} \setminus \{k\}} \sum_{j \in \mathcal{J}} \tilde{a}_{i j; k}(v_{- k}) (v_{i j} - r_{i j}) \\
            + \sum_{i \in \mathcal{I} \setminus \{k\}} \sum_{j \in \mathcal{J}} \tilde{a}_{i j}(v_k, v_{- k}) (v_{i j} - r_{i j})
        \end{multline}
        \begin{multline}
             = \sum_{i \in \mathcal{I}} \sum_{j \in \mathcal{J}} \tilde{a}_{i j}(v_k, v_{- k}) (v_{i j} - r_{i j}) - \\
            \sum_{i \in \mathcal{I} \setminus \{k\}} \sum_{j \in \mathcal{J}} \tilde{a}_{i j; k}(v_{- k}) (v_{i j} - r_{i j}) \\
            + \sum_{j \in \mathcal{J}} \bar{a}_{k j} v_{k j} - \sum_{j \in \mathcal{J}} \bar{a}_{k j} r_{k j} + \sum_{j \in \mathcal{J}} \bar{a}_{k j} \xi_5 ^ k \gamma_{k j}
        \end{multline}
        \begin{multline}
            \geq \sum_{i \in \mathcal{I}} \sum_{j \in \mathcal{J}} \tilde{a}_{i j}(v_k ', v_{- k}) (v_{i j} - r_{i j}) \\
            - \sum_{i \in \mathcal{I} \setminus \{k\}} \sum_{j \in \mathcal{J}} \tilde{a}_{i j; k}(v_{- k}) (v_{i j} - r_{i j}) \\
            + \sum_{j \in \mathcal{J}} \bar{a}_{k j} v_{k j} - \sum_{j \in \mathcal{J}} \bar{a}_{k j} r_{k j} + \sum_{j \in \mathcal{J}} \bar{a}_{k j} \xi_5 ^ k \gamma_{k j}
        \end{multline}
        \begin{multline}\label{eqn:thmIC_last}
            = \sum_{j \in \mathcal{J}} \tilde{a}_{k j}(v_k ', v_{- k}) v_{k j} + \sum_{j \in \mathcal{J}} \bar{a}_{k j} v_{k j} \\
            - \sum_{j \in \mathcal{J}} \tilde{a}_{k j}(v_k ', v_{- k}) r_{k j} - \sum_{j \in \mathcal{J}} \bar{a}_{k j} r_{k j} \\
            - \sum_{i \in \mathcal{I} \setminus \{k\}} \sum_{j \in \mathcal{J}} \tilde{a}_{i j}(v_{- k}) (v_{i j} - r_{i j}) \\
            + \sum_{i \in \mathcal{I} \setminus \{k\}} \sum_{j \in \mathcal{J}} \tilde{a}_{i j}(v_k, v_{- k}) (v_{i j} - r_{i j}),
        \end{multline}
    where the last equality \eqref{eqn:thmIC_last} follows by simple rearrangement using \eqref{eqn:mobility_payment}; therefore, \eqref{eqn:thmIC_last} is equal to $u_k(v_k ', v_{- k})$.
\end{proof}

\begin{theorem}\label{thm:ME}
    The proposed framework ensures that no traveler pays more than their budget for their assignment, i.e., under the pricing mechanism \eqref{eqn:mobility_payment}, for any traveler $i \in \mathcal{I}$, we have $p_i(v) \leq b_i$, for all $v \in \mathcal{V}$.
\end{theorem}

\begin{proof}
    The mobility payment \eqref{eqn:mobility_payment} of any traveler $k$ is
        \begin{multline}\label{thm2:eqnfirst}
            p_k(v) = \sum_{j \in \mathcal{J}} \tilde{a}_{k j}(v) r_{k j} + \sum_{j \in \mathcal{J}} \bar{a}_{k j} r_{k j} - \sum_{j \in \mathcal{J}} \bar{a}_{k j} \xi_5 ^ k \gamma_{k j} \\
            + \sum_{i \in \mathcal{I} \setminus \{k\}} \sum_{j \in \mathcal{J}} \tilde{a}_{i j; k}(v_{- k}) (v_{i j} - r_{i j}) \\
            - \sum_{i \in \mathcal{I}} \sum_{j \in \mathcal{J}} \tilde{a}_{i j}(v) (v_{i j} - r_{i j}).
        \end{multline}
    Next, we formulate the following optimization problem:
        \begin{multline}\label{eqn:estimate:a^v-k}
            (\tilde{a}_{i j; k} (v_{- k}))_{i \in \mathcal{I} \setminus \{k\}, j \in \mathcal{J}} \\
            = \arg \max_{\tilde{a} \in \mathcal{A}_k} \sum_{i \in \mathcal{I} \setminus \{k\}} \sum_{j \in \mathcal{J}} \tilde{a}_{i j} (v_{i j} - r_{i j}),
        \end{multline}
    where $\tilde{a}_{i j; k}(v_{- k}) \in \mathcal{A}_k$ with constraints:
        \begin{align}
            \sum_{i \in \mathcal{I} \setminus \{k\}} \tilde{a}_{i j} & \leq 1 - \sum_{i \in \mathcal{I}} \bar{a}_{i j}, \quad \forall j \in \mathcal{J} \\
            \sum_{j \in \mathcal{J}} \tilde{a}_{i j} \tilde{v}_{i j} & \leq b_i - \sum_{j \in \mathcal{J}} \bar{a}_{i j} r_{i j}, \label{eqn:thmME_one}
        \end{align}
    where \eqref{eqn:thmME_one} holds for all $\tilde{v} \in \mathcal{V}$ and $i \in \mathcal{I} \setminus \{k\}$.
    Thus, we have $(\tilde{a}_{i j; k} (v_{- k}))_{i \in \mathcal{I} \setminus \{k\}, j \in \mathcal{J}} \in \mathcal{A}_k$, which yields
        \begin{align}
            \sum_{i \in \mathcal{I} \setminus \{k\}} \tilde{a}_{i j; k}(v_{- k}) & \leq 1 - \sum_{i \in \mathcal{I}} \bar{a}_{i j}, \quad \forall j \in \mathcal{J}, \\
            \sum_{j \in \mathcal{J}} \tilde{a}_{i j; k}(v_{- k}) \tilde{v}_{i j} & \leq b_i - \sum_{j \in \mathcal{J}} \bar{a}_{i j} r_{i j}, \label{eqn:thmME_two}
        \end{align}
    where \eqref{eqn:thmME_two} holds for all $\tilde{v} \in \mathcal{V}$ and for all $i \in \mathcal{I} \setminus \{k\}$. If we construct an assignment $\alpha = (\alpha_{i j})$ such that
        \begin{equation}
            \alpha_{i j} =
                \begin{cases}
                    \tilde{a}_{i j; k}(v_{- k}), & \quad \forall i \in \mathcal{I} \setminus \{k\}, \quad \forall j \in \mathcal{J}, \\
                    0, & \quad i = k, \quad \forall j \in \mathcal{J},
                \end{cases}
        \end{equation}
    then we have, for all $j \in \mathcal{J}$,
        \begin{align}
            \sum_{i \in \mathcal{I}} \alpha_{i j} & = \sum_{i \in \mathcal{I} \setminus \{k\}} \tilde{a}_{i j; k}(v_{- k}) \leq 1 - \sum_{i \in \mathcal{I}} \bar{a}_{i j}, \\
            \sum_{j \in \mathcal{J}} \alpha_{i j} \tilde{v}_{i j} & = \sum_{j \in \mathcal{J}} \tilde{a}_{i j; k}(v_{- k}) \tilde{v}_{i j} \leq b_i - \sum_{j \in \mathcal{J}} \bar{a}_{i j} r_{i j}, \label{eqn:thmME:three}
        \end{align}
    where \eqref{eqn:thmME:three} holds for all $\tilde{v} \in \mathcal{V}$ and for all $i \in \mathcal{I} \setminus \{k\}$. Similarly, for traveler $k$, it follows from \eqref{Equation:Problem4-ObjectiveFunction}-\eqref{Constraint:Problem4-Third} that
        \begin{equation}
            \sum_{j \in \mathcal{J}} \alpha_{k j} \tilde{v}_{k j} = 0 \leq b_k - \sum_{j \in \mathcal{J}} \bar{a}_{k j} r_{k j}.
        \end{equation}
    Thus, $\alpha \in \mathcal{A}$, and so we have the following inequality:
        \begin{align}\label{eqn:thmME_beforeLast}
            \sum_{i \in \mathcal{I}} \sum_{j \in \mathcal{J}} \tilde{a}_{i j}(v) (v_{i j} - r_{i j}) \geq \sum_{i \in \mathcal{I}} \sum_{j \in \mathcal{J}} \alpha_{i j} (v_{i j} - r_{i j}) \notag \\
            = \sum_{i \in \mathcal{I} \setminus \{k\}} \sum_{j \in \mathcal{J}} \tilde{a}_{i j; k}(v_{- k}) (v_{i j} - r_{i j}).
        \end{align}
    From \eqref{thm2:eqnfirst}, we subtract $\sum_{j \in \mathcal{J}} \tilde{a}_{k j}(v_k, v_{- k}) v_{k j}$ to obtain
        \begin{multline}\label{eqn:thm2_sufficient}
            p_k(v) = \sum_{j \in \mathcal{J}} \tilde{a}_{k j}(v) r_{k j} + \sum_{j \in \mathcal{J}} \bar{a}_{k j} r_{k j} - \sum_{j \in \mathcal{J}} \bar{a}_{k j} \xi_5 ^ k \gamma_{k j} \\
            + \sum_{i \in \mathcal{I} \setminus \{k\}} \sum_{j \in \mathcal{J}} \tilde{a}_{i j; k}(v_{- k}) (v_{i j} - r_{i j}) \\
            - \sum_{i \in \mathcal{I}} \sum_{j \in \mathcal{J}} \tilde{a}_{i j}(v) (v_{i j} - r_{i j}).
        \end{multline}
    From \eqref{eqn:thmME_beforeLast} and \eqref{eqn:thm2_sufficient}, $p_k(v) \leq \sum_{j \in \mathcal{J}} \tilde{a}_{k j}(v) v_{k j} + \sum_{j \in \mathcal{J}} \bar{a}_{k j} r_{k j} - \sum_{j \in \mathcal{J}} \bar{a}_{k j} \xi_5 ^ k \gamma_{k j}$, thus $p_k(v) \leq b_k$.
\end{proof}

\begin{theorem}\label{thm:IR}
    The proposed framework induces all travelers to voluntary participate under the pricing mechanism \eqref{eqn:mobility_payment}, and thus satisfy the last necessary property for mobility equity.
\end{theorem}

\begin{proof}
    By Theorem \ref{thm:IC}, we have 
        \begin{multline}
            u_i(v_k, v_{- k}) = \sum_{i \in \mathcal{I}} \sum_{j \in \mathcal{J}} \tilde{a}_{i j}(v_k, v_{- k}) (v_{i j} - r_{i j}) \\
            - \sum_{i \in \mathcal{I} \setminus \{k\}} \sum_{j \in \mathcal{J}} \tilde{a}_{i j; k}(v_{- k}) (v_{i j} - r_{i j}) \\
            + \sum_{j \in \mathcal{J}} \bar{a}_{k j} v_{k j} - \sum_{j \in \mathcal{J}} \bar{a}_{k j} r_{k j} + \sum_{j \in \mathcal{J}} \bar{a}_{k j} \xi_5 ^ k \gamma_{k j},
        \end{multline}
    which leads to $u_i(v_k, v_{- k}) \geq \sum_{j \in \mathcal{J}} \bar{a}_{k j} v_{k j} - \sum_{j \in \mathcal{J}} \bar{a}_{k j} r_{k j} + \sum_{j \in \mathcal{J}} \bar{a}_{k j} \xi_5 ^ k \gamma_{k j}$, where we have used \eqref{eqn:thmME_beforeLast}. From Lemma \ref{lem:dual_problem} it follows straightforwardly that
        \begin{align}
            \sum_{j \in \mathcal{J}} \bar{a}_{i j} r_{i j} - \sum_{j \in \mathcal{J}} \bar{a}_{i j} v_{i j} ' \leq 0, \quad v_{i j} ' \in \mathcal{V}, \quad \forall i \in \mathcal{I}, \\
            \bar{a}_{i j} r_{i j} = \bar{a}_{i j} v_{i j} ^ {\text{worst}}, \quad \forall i \in \mathcal{I}, \quad \forall j \in \mathcal{J}.
        \end{align}
    Thus, we have
        \begin{align}
            \sum_{j \in \mathcal{J}} \bar{a}_{k j} v_{k j} & \geq \sum_{j \in \mathcal{J}} \bar{a}_{k j} z_{k j} = \sum_{j \in \mathcal{J}} \bar{a}_{k j} r_{k j} \\
            & \geq \sum_{j \in \mathcal{J}} \bar{a}_{k j} r_{k j} - \sum_{j \in \mathcal{J}} \bar{a}_{k j} \xi_5 ^ k \gamma_{k j}.
        \end{align}
    Note that $\bar{a}_{k j}$, $\xi_5 ^ k$, and $\gamma_{i j}$ are non-negative. Thus, we have, for any traveler $i \in \mathcal{I}$, $u_i(v_k, v_{- k}) \geq 0$.
\end{proof}

\section{Discussion}
\label{Section:Conclusion}

\subsection{Implementation}

In this subsection, we outline how our proposed framework can be potentially implemented by considering an example of a major metropolitan area with an extensive road and public transit infrastructure. For example, several key areas in Boston, MA are connected by roads, buses, light rail, and bikes, thus any traveler has access to any of these four modes of transportation. By applying the MaaS concept, a social planner (e.g., central computer) offers travel services (e.g., navigation, location, booking, payment) to all passing travelers at travel hub locations (e.g., train stations with bus stops and taxi waiting line). Information is shared among all travelers via a ``mobility app," which allows them to access the services. Travelers pay for their travel while submitting their individual budget and valuations via a ``preferences" questionnaire within the app. Note that each mode of transportation offers different benefits in utility (e.g., a car is more convenient than a bus and is expected to be in high demand). This justifies our modeling choice of each traveler having valuations for each mobility service. Our design of the payments \eqref{eqn:mobility_payment} guarantee that no traveler has an incentive to misreport these preferences (Theorem \ref{thm:IC}). In addition, travelers are incentivized to use the mobility app multiple times for their travels and interact with each other more than once (Theorem \ref{thm:IR}). Hence, our framework provides an efficient and fair way for travelers to travel using different modes of transportation while competing with many other travelers and pay a fare/toll always within their individual budget using the mobility app (Theorem \ref{thm:ME}).

\subsection{Concluding Remarks}

In this paper, we provided a game-theoretic framework for a multi-modal mobility system where travelers can travel using different modes of transportation and each has a different and unique travel budget. Our goal in this paper was to ensure economic sustainability by maximizing the worst-case revenue of the mobility system under the constraints of mobility equity, which we defined explicitly as truthfulness, voluntary participation, and budget fairness. We proved that our framework ensures budget fairness in the sense that no budget is violated. Under informational asymmetry, we showed that no traveler has an incentive to misreport and they voluntarily participate. Thus, our framework satisfies mobility equity by ensuring access to mobility to all travelers. Ongoing work includes relaxing our assumption of linearity in the utility functions and also investigating our model under the prospect theory behavioral model \cite{chremos2022Prospect}.

%
%
%
%
%

%\addtolength{\textheight}{-7cm}   % This command serves to balance the column lengths

\bibliographystyle{IEEEtran}
\bibliography{references}

\end{document}